\documentclass{article}
\pdfoutput=1
\usepackage[margin=.7in]{geometry}
\usepackage[title]{appendix}

\usepackage{natbib}
\setcitestyle{authoryear,open={(},close={)}} 

\usepackage{amsmath,amsfonts, amssymb, threeparttablex, amsthm}
\newcommand{\E}{\mathbb{E}}
\newcommand{\logit}{\text{logit}}
\newcommand{\N}{\text{N}}
\newcommand{\Ber}{\text{Bernoulli}}
\newcommand\ind{\protect\mathpalette{\protect\independenT}{\perp}}
\def\independenT#1#2{\mathrel{\rlap{$#1#2$}\mkern2mu{#1#2}}}

\def\bSig\mathbf{\Sigma}

\newtheorem{assumption}{Assumption}
\newtheorem{lemma}{Lemma}

\usepackage[figuresright]{rotating}

\title{Identifying and estimating effects of sustained interventions under parallel trends assumptions}

\author{\textbf{Audrey Renson}$^{1,2,*}$, \textbf{Michael G. Hudgens}$^{3}$, \textbf{Alexander P. Keil}$^{1}$,\\ \textbf{Paul N. Zivich}$^{1}$, \textbf{and Allison E. Aiello}$^4$  \\  \bigskip \\\normalsize
$^1$Department of Epidemiology, University of North Carolina at Chapel Hill, Chapel Hill, North Carolina, U.S.A.\\  \normalsize
$^2$Carolina Population Center, University of North Carolina at Chapel Hill, Chapel Hill, North Carolina, U.S.A.\\  \normalsize
$^{3}$Department of Biostatistics, University of North Carolina, 
Chapel Hill, North Carolina, U.S.A. \\  \normalsize
$^4$Columbia Aging Center and Department of Epidemiology, Mailman School of Public Health,
\\ \normalsize Columbia University, New York, New York, U.S.A. \\ \normalsize
$^*$ Email: audrey.o.renson@gmail.com}

\begin{document}

\maketitle

\begin{abstract}
Many research questions in public health and medicine concern sustained interventions in populations defined by substantive priorities. Existing methods to answer such questions typically require a measured covariate set sufficient to control confounding, which can be questionable in observational studies. Differences-in-differences relies instead on the parallel trends assumption, allowing for some types of time-invariant unmeasured confounding. However, most existing difference-in-differences implementations are limited to point treatments in restricted subpopulations. We derive identification results for population effects of sustained treatments under parallel trends assumptions. In particular, in settings where all individuals begin follow-up with exposure status consistent with the treatment plan of interest but may deviate at later times, a version of Robins' g-formula identifies the intervention-specific mean under SUTVA, positivity, and parallel trends. We develop consistent asymptotically normal estimators based on inverse-probability weighting, outcome regression, and a double robust estimator based on targeted maximum likelihood. Simulation studies  confirm theoretical results and support the use of the proposed estimators at realistic sample sizes. As an example, the methods are used to estimate the effect of a hypothetical federal stay-at-home order on all-cause mortality during the COVID-19 pandemic in spring 2020 in the United States.
\end{abstract}

%
%


\section{Introduction}
Many epidemiologic  and other empirical studies concern the effects of sustained treatment strategies on population average outcomes over time. A sustained treatment or intervention is one that sets values of a time-varying exposure via a predetermined plan or algorithm.  For example, clinical studies are often concerned with optimal dosing plans for therapeutic drugs, and policy investigations are often concerned with policies that determine exposure distributions repeatedly over time for the population residing in a jurisdiction.

Existing approaches to estimating  effects of sustained interventions in well-defined populations include g-computation \citep{Robins1986}, inverse probability of treatment weighted (IPTW) marginal structural models \citep{Robins2000, Robins2000a}, g-estimation of structural nested models \citep{robins1989analysis}, and double robust methods such as augmented IPTW \citep{Bang2005} and targeted maximum likelihood \citep{Vanderlaan2012}. Importantly, these approaches base causal identification on a sequential version of exchangeability \citep{Robins1986}, also known as sequential ignorability or no unmeasured confounders \citep{Robins2000}. Sequential exchangeability posits that the potential outcomes are independent of treatment, given the history of some set of measured (possibly time-varying) covariates and treatment; this assumption is unverifiable and can be implausible in many settings. For example, individuals may select medical treatments based on unmeasured risk factors, and public policies are decided in highly complex political contexts that may influence health. 

In contrast, difference-in-differences (DID) methods typically base identification on parallel trends assumptions rather than sequential exchangeability \citep{Ashenfelter1985, roth2022s}. Parallel trends assumptions posit that time trends in average potential outcomes are independent of the observed treatment \citep{Ashenfelter1985, Marcus2020}.  DID methods typically focus on the average treatment effect in the treated for a treatment occurring at a single time point, although recently extensions have considered certain types of sustained treatment regimes \citep[for a review, see][]{roth2022s}. In particular, \citet{callaway2021difference} and \citet{ chaise2020hetero, chaise2021several} consider effects conditional on each observed treatment path in a monotonic (i.e., staggered adoption) setting, meaning that values of the observed time-varying treatment can either increase or decrease over time, but not both. Somewhat more generally, \citet{chaise2021inter} consider interventions fixing a (possibly non-monotonic) binary or ordinal exposure to its baseline status, focusing on unconditional cost-effectiveness ratios and outcome regression estimators. Relevant to the present work, these recent DID developments have included doubly robust estimators \citep{callaway2021difference, sant2020}.

This paper considers a more general setting and different estimands than previous DID implementations, focusing on marginal effects of general sustained treatment strategies under parallel trends assumptions. The remainder of this paper is organized as follows.  Section 2 introduces notation and the assumed data structure. Section 3 presents new identification results for an intervention-specific mean under parallel trends, where identifying formulas are modifications of Robins' (1986) g-computation algorithm formula (g-formula). Section 4 then presents consistent and asymptotically normal (CAN) estimators based on inverse probability weighting, outcome regression, and a double robust estimator that combines both. Section 5 presents simulation results that support the identification result and theoretical large sample properties of the proposed estimators. Section 6 presents an example estimating the number of lives that would have been saved by a U.S. federal stay-at-home order during the COVID-19 pandemic in spring 2020. Section 7 considers sensitivity analysis for violations of parallel trends and application of the poposed approach to dynamic treatment regimens. Section 8 concludes.

\section{Preliminaries}\label{sec2}
    \subsection{Data}
    Suppose data $O_{it}=\{W_{it},A_{it},Y_{it}\}$ are observed on $i=1,2,...,n$ individuals (or units) at time points $t=0,1,...,\tau$, where: $W_{it}$ are (possibly vector-valued) covariates; $A_{it}$ are discrete, possibly multivariate treatments realized after $W_{it}$; and $Y_{it}$ are outcomes realized after $A_{it}$, all measured without error.  Denote history of a variable with overbars, e.g., $\overline A_{it}=(A_{i0}, A_{i1}, ..., A_{it})$, with $\overline A_i\equiv\overline A_{i\tau}$ and $A_{ik}=\{\emptyset\}$ for $k<0$ by convention.   Upper case is used throughout to refer to random variables, lower case refers to specific realizations, and scripts refer to the support. The $i$ subscript is omitted unless needed to resolve ambiguity. Throughout, it is assumed that $\overline O_i\equiv\{\overline W_i, \overline A_i, \overline Y_i\}$ $(i=1,2,...,n)$ represent independent and identically distributed (iid) draws from a relevant target population.
    
    Assume the data come from a \textit{staggered discontinuation design}, defined as follows. Suppose the target estimand is $\E\{Y_t(\overline a^*)\}$ where $\overline a^*$ denotes the treatment strategy or intervention plan of interest, and $Y_t(\overline a^*)$ denotes a potential outcome; i.e., the value $Y_t$ would take under the intervention setting $\overline A_t=\overline a_t^*$. The approach in this paper requires that in the observed data distribution, $\Pr(A_0=a_0^*)=1$, or in other words that $A_{i0}=a_{0}^*$ for all $i$. Such a scenario is said to be a staggered discontinuation design with respect to the treatment plan  $\overline a^*$ of interest, because units begin follow-up under the treatment plan but may discontinue at later points in a staggered way. Note that we do not require monotonic treatment assignment, in contrast to recent DID papers \citep[e.g.,][]{Goodman-bacon2018,callaway2021difference}.
    
    \subsection{Motivating example}\label{sec:ex_intr}
    
    Consider the question, ``what effects on all-cause mortality would a U.S. federal stay-at-home order have had in spring 2020 during the COVID-19 pandemic?'' Let $Y_{it}$ be a binary indicator that individual $i$ died during week $t$, $t=0,1,...,11$, measured as weeks since  April 6, 2020. Let $A_{it}$ be a binary indicator that the state in which individual $i$ was living during week $t$ was under a state-level stay-at-home or shelter-in-place order. Suppose it is of interest to estimate $\E\{Y_t(\overline 1)\}-\E(Y_t),$ the difference in U.S. mortality rates under a hypothetical federal stay-at-home order vs. under the observed treatment trajectory (i.e., the ``natural course''). As of April 6, 43/50 U.S. states were under stay-at-home orders, which were discontinued at times ranging from late April to late June, with the exception of California which continued through December (Figure \ref{fig:sah}). Thus, the observed treatment trajectories give rise to a staggered discontinuation design with respect to the treatment plan $\overline a^*=\overline 1$ setting everyone to remain under stay-at-home order in those 43 states. The methods developed below can be used to draw inference about what would have happened had such a policy been implemented.

	\begin{figure}[h]
		\centerline{\includegraphics{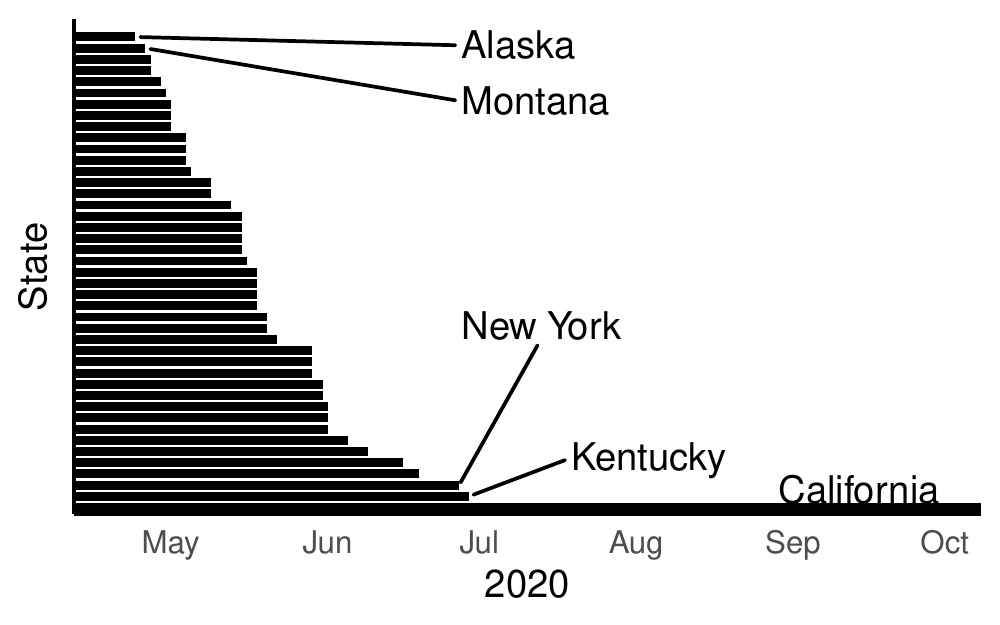}}
		\caption{Dates of state-issued stay-at-home orders in U.S. states during the COVID-19 pandemic in 2020\label{fig:sah}}
	\end{figure}

\section{Identification}\label{sec:id}
In this section we consider identification, given data $\overline O$, of the quantity $\mu_t\equiv\E\{Y_t(\overline a^*)\}$, the mean outcome at time $t$, under the intervention to set all individuals to $\overline A_i=\overline a^*$.  Throughout, it is assumed that interest lies in only one intervention $\overline a^*$. Note that $\mu_t$ depends on $\overline a^*$, which is left implicit for notational simplicity.  Consider the following assumptions:
    \begin{assumption}\label{asn:sutva} (Stable Unit Treatment Value Assumption [SUTVA]): If $\overline A_{it} = \overline a_t^*$, then $Y_{it}=Y_{it}(\overline a_t^*)$ for $t \in \{0,1,...\tau\}$.
    \end{assumption}
    In words, Assumption \ref{asn:sutva} requires that an individual's potential outcomes are not affected by other individuals' treatments (i.e., no interference), allowing us to index potential outcomes by each individual's treatments alone. Assumption \ref{asn:sutva} also requires that the treatment $\overline{a}^*$ is defined precisely enough so that, for individuals whose observed treatment equals $\overline{a}_t^*$,  observed outcomes can stand in for counterfactual outcomes under a hypothetical intervention $\overline{a}_t^*$. Implicit in Assumption \ref{asn:sutva} is that the future cannot affect the past, or $Y_{it}(\overline a_k^*)=Y_{it}(\overline a_t^*)$ for $k>t$. Thus, SUTVA would be violated, for example, if individuals were able to anticipate future treatments and change their behavior. This particular violation has typically been addressed under a separate ``no anticipation'' assumption in the DID literature \citep[e.g., ][]{callaway2021difference}. SUTVA is an unverifiable assumption in observational studies, though it is sometimes possible to test for the presence of interference \citep{Halloran2016}.
    \begin{assumption}\label{asn:pos} (Positivity): If $f(\overline w_t|\overline A_{t-1}=\overline a_{t-1}^*)>0,$ then $f(\overline a_t^*|\overline W_t=\overline w_t,\overline A_{t-1}=\overline  a_{t-1}^*)>0$, for $\overline w_t \in \mathcal{\overline{W}}_t; t\in\{1,2,...,\tau\}$.
    \end{assumption}
    Here and throughout, $f(x|\cdot)$ refers to a conditional density if $X$ is continuous, and a conditional probability mass function if $X$ is discrete. Assumption \ref{asn:pos} requires that units whose treatment history up to time $t-1$ is consistent with the regime in question ($\overline a^*)$ have positive probability of remaining under treatment plan $\overline a^*$ at time $t$. Positivity can sometimes be a verifiable assumption. In particular, if both $\overline W$ and $\overline A$ are low-dimensional and discrete, then among units with $\overline A_{t-1}=\overline a_{t-1}^*,$ if one observes units who remain under $\overline a^*$ at time $t$ in every stratum of $\overline W_t,$ this implies positivity in the population with probability 1 (but not necessarily the reverse).
    \begin{assumption}\label{asn:pt} (Parallel trends): For $t \in \{1,2,...,\tau\}, k\leq t:$
    	\[
    	\E\{Y_t(\overline a^*)-Y_{t-1}(\overline a^*)|\overline W_k, \overline A_{k-1}=\overline a_{k-1}^*\}=\E\{Y_t(\overline a^*)-Y_{t-1}(\overline a^*)|\overline W_k, \overline A_k=\overline a_k^*\} \]
    \end{assumption}
    In words, Assumption \ref{asn:pt} states that, among individuals whose treatment status is consistent with the intervention $\overline a^*$ up to time $k-1$, had (counter to fact) all individuals followed the intervention through time $t$, trends would have been parallel for those who do and do not follow the intervention at time $k$ but have equal covariate histories. This assumption is very similar to those adopted in papers describing event study and staggered adoption DID designs such as \citet{Goodman-bacon2018} and \citet{callaway2021difference}, and nearly identical to Assumption 12 in \citet{chaise2021inter}. The latter differs by considering only certain types of treatment regimes and presuming a linear relation between covariates and counterfactual trends. Parallel trends is unverifiable, though closely related conditions can often be checked \citep{Roth2019}. 
    Note that $W_{it}$ may include prior outcomes $Y_{im}$ for $m<t$. However, if $Y_{i,t-1}$ is included in $W_{it}$, then the parallel trends assumption is equivalent to sequential exchangeability, in which case existing causal inference methods for observational data with a longitudinal exposure can be used \citep[e.g.,][]{Robins1986, Robins2000, Vanderlaan2012}. 
    
     The following Lemma presents the main identification results in this paper, which show that Assumptions \ref{asn:sutva}-\ref{asn:pt} are sufficient to equate $\mu_t$ to a function of the observed data distribution. 
    \begin{lemma}\label{lem:gform} (Parallel trends g-formula) Define the functional (i.e., statistical parameter)
    $\psi_t \equiv \E(Y_0)+\sum_{k=1}^t \int \E(Y_k-Y_{k-1}|\overline W_k =\overline w_k,  \overline A_k =\overline a_k^*)\prod_{m=0}^k dF(w_m|\overline w_{m-1},  \overline a_{m-1}^*)$. Under a staggered discontinuation design and if Assumptions \ref{asn:sutva}-\ref{asn:pt} hold, then $\psi_t = \mu_t$.
    \end{lemma}
     Here and throughout, $F(\cdot|\cdot)$ refers to a conditional cumulative distribution function. Lemma \ref{lem:gform} states that the target causal quantity $\mu_t$ is identified by the parameter $\psi_t$. The parameter $\psi_t$ is referred to as the \textit{parallel trends g-formula} because it represents a modification of the usual g-formula (the dependence of $\psi_t$ on $\overline a^*$ is also left implicit).  A formal proof of Lemma \ref{lem:gform} by induction is presented in Appendix A. Here we give a less formal explanation to build intuition. We have:
    \begin{align*}
    	\mu_t\equiv\E\{Y_t(\overline a^*)\} &= \E\{Y_0(\overline a^*)\} + \sum_{k=1}^t \E\{Y_k(\overline a^*)-Y_{k-1}(\overline a^*)\} \\
    	&= \E\{Y_0(\overline a^*)\} + \sum_{k=1}^t \E[\E\{Y_k(\overline a^*)-Y_{k-1}(\overline a^*)|W_0\}] \\
    	&= \E\{Y_0(\overline a^*)\} + \sum_{k=1}^t \E[\E\{Y_k(\overline a^*)-Y_{k-1}(\overline a^*)|A_0=a_0^*, W_0\}]
    \end{align*}
    where the first equality follows by adding and subtracting constants, the second by iterated expectation, and the third by Assumption \ref{asn:pt}. Repeatedly applying iterated expectation and Assumption \ref{asn:pt}, we have:
    \begin{align*}
    	\mu_t &= \E\{Y_0(\overline a^*)\} + \sum_{k=1}^t \E(\cdot\cdot\cdot\E[\E\{Y_k(\overline a^*)-Y_{k-1}(\overline a^*)|\overline{A}_k=\overline a^*_k, \overline W_k\}|\overline A_{k-1}=\overline a^*_{k-1}, \overline W_{k-1} ]\cdot\cdot\cdot|A_0=a_0^*, W_0)\\
    	&=\E(Y_0) + \sum_{k=1}^t \E[\cdot\cdot\cdot\E\{\E(Y_k-Y_{k-1}|\overline A_k=\overline a^*_k, \overline W_k)|\overline A_{k-1}=\overline a^*_{k-1}, \overline W_{k-1} \}\cdot\cdot\cdot|A_0=a_0^*, W_0]\\
    	&=\psi_t
    \end{align*}
    where the second equality follows from Assumption \ref{asn:sutva} and the last equality from iterated expectation. 

\section{Estimators}\label{sec:est}
This section presents estimators for the statistical parameter $\psi_t$, which equals the target quantity $\mu_t$ under the above stated assumptions. The estimators in this section utilize existing estimators of $\mu_t$ under sequential exchangeability rather than parallel trends, all of which are CAN estimators of the g-formula by virtue of being solutions to unbiased estimating equations \citep{Stefanski2002}. Since $\psi_t$ is a continuous function of several g-formulas, the same function applied to estimators of those g-formulas is a CAN estimator for  $\psi_t$. The remainder of this section formalizes this logic and gives examples of specific estimators that function in this capacity. The estimators presented in this section are provided in an R package (see Supporting Information).
    
    \subsection{General form}\label{sec:est_gen}
    Here we derive a general form of a CAN estimator for the target statistical parameter, $\psi_t$. First, define:
    \begin{equation}\label{eq:gform}
    \phi_{j,k}=\int \E(Y_j|\overline A_k=\overline a_k^*, \overline W_k=\overline w_k)\prod_{m=0}^k dF(w_m|\overline  W_{m-1}=\overline w_{m-1}, \overline A_{m-1}=\overline a_{m-1}^*)
    \end{equation}
    Equation (\ref{eq:gform}) is the g-formula, developed in the context of identifying parameters like $\mu_t$ under sequential exchangeability. Here, sequential exchangeability is not assumed, and therefore $\phi_{j,k}$ is not interpretable as a causal parameter; instead existing estimators of the statistical parameter $\phi_{j,k}$ are used to assemble estimators of $\psi_t$ (which equals the causal parameter $\mu_t$ under Assumptions \ref{asn:sutva}-\ref{asn:pt}) by noting that $\psi_t = \phi_{0,0} + \sum_{k=1}^t \big(\phi_{k,k}-\phi_{k-1,k}\big)$. Next, suppose there is an estimator $\widehat \phi_{jk}$ of $\phi_{jk}$ which is the solution to an unbiased estimating function $d_{\phi_{jk}}(O;\phi_{jk})$; i.e., $0=\E\{d_{\phi_{jk}}(O;\widehat\phi_{jk})\}.$ Let $\boldsymbol\phi=(\phi_{0,0}, \phi_{0,1},...,\phi_{t-1, t}, \phi_{t,t}).$ Then simply define
    \begin{equation}\label{eq:ee_psi}
    d_{\psi_t}(O; \boldsymbol\phi,\psi_t)=\phi_{0,0}+\sum_{k=1}^t\big(\phi_{k,k}-\phi_{k-1, k}\big) -\psi_t    
    \end{equation}
    Clearly, $\E\{d_{\psi_t}(O; \boldsymbol\phi,\psi_t)\}=0$, indicating that an estimator $(\boldsymbol{\widehat\phi},\widehat\psi_t)$ that jointly solves $0=\sum_i d_{\psi_t}(O_i; \boldsymbol{\widehat\phi}, \widehat\psi_t)$ will yield a CAN estimator $\widehat\psi_t$ for $\psi_t$. 
    
    The following subsections show how different options for $d_{\phi_{jk}}(O;\phi_{jk})$ (including IPTW, g-computation, and targeted maximum likelihood estimation [TMLE]) can be constructed and stacked with (\ref{eq:ee_psi}) to form estimators of $\psi_t$ that inherit desirable properties (e.g., consistency, asymptotic normality, double robustness). 
    
    \subsection{Inverse probability of treatment weighted (IPTW) estimator}
   Define stabilized inverse probability of treatment weights (IPTWs) as 
   $\pi_k(\overline a; \overline W)= \linebreak \prod_{m=1}^k f(a_m|\overline a_{m-1})/\prod_{m=1}^k f(a_m|\overline a_{m-1}, \overline W_m)$.     \citet{Robins2000} (Lemma 1.1) showed that $\phi_{j,k}=c_{jk}(\overline a^*_k),$ where $c_{jk}(\cdot)$ is the unique function such that $\E[q_{jk}(\overline A_k)\{Y_j - c_{jk}(\overline A_k)\}\pi_k(\overline A; \overline W)]=0$ for all functions $q_{jk}(\overline A_k)$ where the expectation exists. The equation \\$\E[q_{jk}(\overline A_k)\{Y_j - c_{jk}(\overline A_k)\}\pi_k(\overline A_k; \overline W)]=0$ defines a regression of $Y_j$ on $\overline A_k$ weighted by $\pi_k(\overline A; \overline W)$. In the context of sequential exchangeability, estimators based on this weighted regression formulation are called IPTW-marginal structural model estimators and given a causal interpretation \citep{robins1992estimating, Robins2000}. However, in the context of this paper, sequential exchangeability is not assumed and so the weighted regression equation does not have a causal interpretation on its own. Instead, the above result together with results from Section \ref{sec:est_gen}, implies that an IPTW estimator of $\psi_t$ can be formed by using a linear combination of IPTW marginal structural model estimators for $\phi_{j,k}$. 
    
    For simplicity of presentation, assume that for $k=0,1,...,t$, $f(a_k|\overline a_{k-1})$ and $f(a_k|\overline a_{k-1}, \overline w_k)$ are known up to a finite dimensional parameter. That is, define $g_{0, k}(\overline a, \emptyset) \equiv \prod_{s=0}^k f(a_s|\overline a_{s-1})$ and $g_{0, k}(\overline a, \overline w) \equiv \prod_{s=0}^k f(a_s|\overline a_{s-1}, \overline w_s)$, and say we are willing to assume that $g_{0, k}(\overline a, \emptyset)$ is uniquely determined by the parametric model $g_k(\overline a, \emptyset; \alpha_0)$, and similarly that $g_{0, k}(\overline a, \overline w) = g_k(\overline a, \overline w; \alpha_1)$, where $\alpha_0$ and $\alpha_1$ are finite dimensional parameter vectors. Say we have an estimator that solves an unbiased estimating equation for $(\alpha_0, \alpha_1)$. For example, $g_k(\overline a, \emptyset; \alpha_0)$ and $g_k(\overline a, \overline w; \alpha_1)$ may consist of generalized linear models with parameters estimated by maximum likelihood. Specify $c_{jk}(\overline A_k)$ as some appropriate functional form for the expected value of $Y_j$ conditional on $\overline A_k$ in the weighted data distribution, such as $c_{jk}(\overline A_k) = \gamma_{0jk} + \gamma_{1jk}I(\overline A_k=\overline a_k^*)$  (i.e., leaving the model unrestricted when $\overline A_k=\overline a_k^*$).  Then an estimator $\widehat\phi_{j,k}^{IPTW}$ that solves $0=c_{jk}(\overline a^*_k) - \phi_{j,k}$ is CAN for $\phi_{j,k}$ if $c_{jk}(\overline A_k), g_k(\overline A, \emptyset; \alpha_0)$ and $g_k(\overline A, \overline w; \alpha_1)$ are correctly specified. Finally, stack the score equations for $g_k(\overline a, \emptyset; \alpha_0)$ and $g_k(\overline a, \overline w; \alpha_1)$, along with $0=c_{jk}(\overline a^*_k) - \phi_{j,k}$ $(k=0,1,...,t; j=k-1,k)$ and equation (\ref{eq:ee_psi}) to yield an estimator for $\psi_t$, say $\widehat\psi_t^{IPTW}.$  

In other words, the IPTW estimator for the target parameter of interest is $\widehat\psi_t^{IPTW}=\widehat\phi_{0,0}^{IPTW} + \sum_{k=1}^t\big(\widehat\phi_{k,k}^{IPTW}-\widehat\phi_{k-1,k}^{IPTW}\big)$, where $\widehat\phi_{j,k}^{IPTW}$ are estimators of each appropriate g-formula parameter based on an IPTW model. Note that under our assumption set, $\widehat\phi_{j,k}^{IPTW}$ are not estimators of causal quantities in and of themselves, but simply functions of the observed data distribution that may be assembled appropriately to form the causal estimator $\widehat\psi_t^{IPTW}$. Clearly, $\widehat\psi_t^{IPTW}$ solves an estimating equation that is unbiased if $g_k(\overline A, \emptyset; \alpha_0), g_k(\overline A, \overline w; \alpha_1)$ and $c_{jk}(\overline A_k)$ are all correctly specified, implying $\widehat\psi_t^{IPTW}$ is CAN for $\psi_t$ under the same conditions. However, IPTW estimators are known to be inefficient and $\widehat\psi_t^{IPTW}$ may similarly inherit this property. The following subsections present estimators that may improve on efficiency relative to IPTW.
\subsection{Iterated conditional expectation (ICE) estimator}
\citet{Bang2005} describe an estimator of $\phi_{j,k}$ based on the following iterated conditional expectation (ICE) representation:
    \begin{equation}\label{eq:ice}
    \phi_{j,k} = \E(\E[\cdot\cdot\cdot\E\{\E(Y_j|\overline A_k=\overline a_k^*,\overline W_k)|\overline A_{k-1}=\overline a_{k-1}^*,\overline W_{k-1}\}\cdot\cdot\cdot|\overline A_1=\overline a_1^*,\overline W_1]|A_0=a_0^*, W_0)    
    \end{equation}
which can equivalently be written as $\phi_{j,k} = \E\{Q_0^{j,k,0}(\overline a^*)\}$ where, for $m=0,1,...,k,$ $Q_0^{j,k,m}(\overline a^*)=\E\{Q_0^{j,k,m+1}(\overline a^*)|\overline A_m=\overline a_m^*, \overline W_m\}$ and $Q_0^{j,k,k+1}(\overline a^*)=Y_j.$ 

An estimator of $\psi_t$ can then be formulated based on this representation. For simplicity, say we are willing to assume that $Q_0^{j,k,m}(\overline a^*)$ are known up to a finite dimensional parameter for $m=0,1,...,k$. That is, assume $Q_0^{j,k,m}(\overline a^*)=Q^{j,k,m}(\overline a^*; \beta_m)$, where $\beta_m$ for $m=0,1,...,k,$ are finite dimensional parameters. For example, $Q^{j,k,m}(\overline a^*; \beta_m)$ may be a generalized linear model with parameters $ \beta_m$. Say we are in possession of an unbiased estimating function $d_{j,k,m}\{O, Q^{j,k,m+1}(\overline a^*; \beta_{m+1}); \beta_m\}$ for $\beta_m$. For example, if maximum likelihood is used, then $d_{j,k,m}\{O, Q^{j,k,m+1}(\overline a^*; \beta_{m+1}); \beta_m\}$ is the vector of first derivatives of the model log-likelihood with respect to $\beta_m$.  Note that including $Q^{j,k,m+1}(\overline a^*; \beta_{m+1})$ as an argument to the estimating function makes explicit the nested nature of the iterated expectations being modeled. The ICE estimator of $\phi_{j,k}$ is then defined \citep{Bang2005} as the solution $\widehat\phi_{j,k}^{ICE}$ to $\boldsymbol 0=\sum_{i=1}^n d_{j,k}(O_i; \phi_{j,k})$ where 
\[
d_{j,k}(O; \phi_{j,k}) = \begin{pmatrix}
    	d_{j,k,k}(O; \beta_k) \\
    	d_{j,k,k-1}\{O, Q^{j,k,k}(\overline a^*; \beta_k); \beta_{k-1}\} \\
    	\vdots \\
    	d_{j,k,0}\{O, Q^{j,k,1}(\overline a^*; \beta_1); \beta_0\} \\
    	Q^{j,k,0}(\overline a^*; \beta_0) -  \phi_{j,k}
    \end{pmatrix}
\]
Then, simply stack $d_{j,k}(O; \phi_{j,k})$ with (\ref{eq:ee_psi}) to yield an estimator  $\widehat\psi_t^{ICE}$ for $\psi_t$. In other words, the ICE estimator of the target parameter is $\widehat\psi_t^{ICE}=\widehat\phi_{0,0}^{ICE} + \sum_{k=1}^t\big(\widehat\phi_{k,k}^{ICE}-\widehat\phi_{k-1,k}^{ICE}\big)$, where each $\widehat\phi_{j,k}^{ICE}$ is an estimator of the corresponding g-formula parameter based on ICE g-computation. Clearly, $\widehat\psi_t^{ICE}$ solves an unbiased estimating equation whenever all the iterated outcome models $\{Q^{j,k,m}(\overline a^*; \beta_m): k=0,1,...,t; j=k,k-1; m=0,1,...,k+1\}$ are correctly specified. Estimators of $\phi_{j,k}$ based on outcome regression generally have smaller asymptotic variance that IPTW estimators, and $\widehat\psi_t^{ICE}$ may inherit this property. 
    
    \subsection{Doubly robust targeted maximum likelihood estimator (TMLE)}
    IPTW estimators are only guaranteed to be CAN if the treatment models are correctly specified, and ICE estimators are only guaranteed to be CAN if all the outcome models are correctly specified. Doubly robust estimators are CAN if either the outcome or treatment models are correct (but not necessarily both), which is an advantage because one is rarely certain that models are correctly specified. 
        
    Doubly robust estimators of $\phi_{j,k}$ generally consist of augmenting the ICE algorithm by including predicted values from the treatment models used to construct IPTWs in some way. Such estimators are called semiparametric efficient if they solve the estimating equation corresponding to the following efficient influence curve \citep{Vanderlaan2012, Tran2019}:
    \begin{equation}\label{eq:eic}
    \sum_{m=0}^k \frac{I(\overline A_m=\overline a_m^*)}{g_{0, m}(\overline a^*, \overline w)}\{Q_0^{j,k,m+1}(\overline a^*)-Q_0^{j,k,m}(\overline a^*)\} + Q_0^{j,k,0}(\overline a^*) - \phi_{j,k}
    \end{equation}
    with $g_{0, m}(\overline a^*, \overline w)$ and $Q_0^{j,k,m}$ defined as in previous sections. Many estimators correspond to this efficient influence curve, meaning they all have the smallest asymptotic variance of any regular asymptotically linear estimator in this class \citep{Bang2005,Vanderlaan2012}. We present one such example of a targeted maximum likelihood estimator (TMLE) which may outperform others in finite samples \citep{Tran2019}. 
    
    First consider the TMLE of $\phi_{j,k}.$ For simplicity, assume that outcome models $\{Q_0^{j,k,m}(\overline a^*): m=0,1,...,k\}$ and treatment models $\{g_{0, m}(\overline a^*): m=0,1,...,k\}$ are known up to a finite dimensional parameter. That is, assume $g_{0, m}(\overline a^*)=g_m(\overline a^*; \alpha_m)$ and $Q_0^{j,k,m}(\overline a^*)=Q^{j,k,m}(\overline a^*; \beta_m)$, where $\alpha_m$ and $\beta_m$ are finite dimensional parameters, $m=0,1,...,k$. Then proceed as follows:
    
    \begin{enumerate}
    	\item For $m=0,1,...,k$, estimate $\alpha_m$, for example using maximum likelihood. Denote estimators of $\alpha_m$ as $\widehat\alpha_m$ and corresponding estimators of $g_m(\overline a^*; \alpha_m)$ as $g_m(\overline a^*; \widehat\alpha_m)$.
    	\item For $m=k,$ estimate $\beta_{m}$, for example using maximum likelihood, denoting this estimator $\widehat\beta_m$. Calculate $Q^{j,k,m}_i(\overline a^*; \widehat\beta_m)$ for each unit $i$ and denote this estimator $\widehat Q^{j,k,m}_i(\overline a^*)$. Note that these are model predictions that implicitly depend on the data, and so vary across units $i$.
    	\item Also for $m=k,$ update the initial fit $\widehat Q^{j,k,m}_i(\overline a^*)$ by fitting a new model, defined as $h\{Q^{j,k,m,*}_i(\overline a^*)\} = h\{\widehat Q^{j,k,m}_i(\overline a^*)\} + \epsilon_{j,k,m}$, where $h(\cdot)$ is an appropriate link function, $\epsilon_{j,k,m}$ is an intercept, and $Q^{j,k,m,*}_i(\overline a^*)$ are conditional expectations under the updated model. Note, the response variable in this model is $Q^{j,k,k+1}_i(\overline a^*)=Y_j.$ The logit link is recommended to ensure the estimator respects bounds implied by the data (if $Y_j$ is not bounded by $(0,1)$, it will need to be appropriately transformed for the logit function to be defined) \citep{Vanderlaan2012}. Estimators $\widehat Q^{j,k,m,*}_i(\overline a^*)$ for the updated fit are found by maximizing an appropriate weighted likelihood with weights $I(\overline A_{im}=\overline a_m^*)/g_m(\overline a^*; \widehat\alpha_m)$.
    	\item Repeat steps 2-3, estimating $Q^{j,k,m}(\overline a^*; \beta_m)$ and $Q^{j,k,m,*}(\overline a^*)$ for $m=k-1,k-2,...,0.$
    	\item The TMLE for $\phi_{j,k}$ is then defined as $\widehat\phi_{j,k}^{TMLE}=n^{-1}\sum_{i=1}^n\widehat Q^{j,k,0,*}_i(\overline a^*)$.
    \end{enumerate}
    Then, the TMLE for $\psi_t$ is defined as $\widehat\psi_t^{TMLE}=\widehat\phi_{0,0}^{TMLE} + \sum_{k=1}^t\big(\widehat\phi_{k,k}^{TMLE}-\widehat\phi_{k-1,k}^{TMLE}\big)$. Since $\widehat\phi_{j,k}^{TMLE}$ solves the estimating equation corresponding to the efficient influence curve (\ref{eq:eic}), it will be CAN for $\phi_{j,k}$ so long as either (i) the set of outcome models \{$Q^{j,k,m}(\overline a^*; \beta_m):m=0,1,...,k\}$ are correctly specified, or (ii) the set of treatment models $\{g_m(\overline a^*; \alpha_m): m=0,1,...,k\}$ are correctly specified, but it is not necessary that both be correct. Therefore, if one of these two conditions holds for all $k=0,1,...,t$ and $j=k-1,k$, then $\widehat\psi_t^{TMLE}$ will be CAN for $\psi_t$. The double robustness property carries through to $\widehat\psi_t^{TMLE}$ by virtue of the fact that the estimating equation in (\ref{eq:ee_psi}) is unbiased if the estimating equations for all the $\phi_{j,k}$ are unbiased, which is the case for $\widehat\phi_{j,k}^{TMLE}$ under conditions (i) or (ii) above.

    \section{Simulation study}
    A simulation study was conducted to evaluate the finite sample performance of the IPTW, ICE, and TMLE estimators  described in Section \ref{sec:est} when Assumptions \ref{asn:sutva}-\ref{asn:pt} hold and all models were correctly specified. The TMLE estimator was also evaluated under misspecification of either the treatment or outcome model. Code for the simulation is provided in an R package (see Supporting Information).
    
    \subsection{Data generating distribution}\label{sec:sim_dgd}
    Data $O_{it}=\{W_{it}=(W_{it1},W_{it2}), A_{it}, Y_{it}\}; i=1,2,...,n; t=0,...,5$ were generated from the distributions $U_{i0} \sim \Ber\{\logit^{-1}(\omega_0)\}$; $W_{it1} \sim \Ber \{\logit^{-1}(\alpha_{0t} + \alpha_{1t}A_{i,t-1})\}$; $W_{it2} \sim \N(\gamma_{0t}+\gamma_{1t}A_{i,t-1}, 1)$; $A_{it}|A_{i,t-1}=0 \sim \Ber \{\logit^{-1}(\delta_{0t} + \delta_{1t}U_{i0} + \delta_{2t}W_{it1}+ \delta_{3t}W_{it2}+ \delta_{4t}W_{it2}^2)\}$; and $Y_{it} \sim N(\beta_{0t} + \beta_{1t}W_{it1}+ \beta_{2t}W_{it2}+ \beta_{3t}W_{it2}^2 + \beta_{4t}A_{it} + \theta U_{i0} , 1)$; with $A_{i0}=0$ and $A_{it}=1$ if $A_{i,t-1}=1$. The monotonic treatment assignment for $A_{it}$ is not necessary, but simplifies analysis. In all analyses, $U_{i0}$ is treated as unmeasured, but all other variables are observed. Parallel trends hold in this setup, as one sufficient set of conditions for parallel trends (proven in  Appendix B) is that the only unmeasured variables (here, $U_{i0})$ are time-invariant, do not affect time-varying covariates, and enter the outcome model linearly with constant coefficient over time (here, $\theta$ is constant over $t$).
    
    We simulated 1,000 datasets each for sample sizes $n=$1,000, 10,000, and 100,000. All parameters were generated from a $N(0.2, 1)$ distribution, with the same values for each parameter used across all simulation runs. The target parameter was $\mu_5 = \E\{Y_5(\overline 0)\}=-3.98$, the mean outcome at end of follow-up, had everyone remained untreated. The true difference compared to the natural course was $\E\{Y_5(\overline 0)\}-\E(Y_5)=-3.98 - (-3.88) = -0.10.$ 
    \subsection{Estimator implementation}
    For each simulated dataset, the estimators $\widehat\mu_t^{IPTW},$ $\widehat\mu_t^{ICE},$ and $\widehat\mu_t^{TMLE}$ were calculated using correctly specified generalized linear regression models estimated using maximum likelihood. Additionally, $\widehat\mu_t^{IPTW}$ was calculated with treatment models misspecified, $\widehat\mu_t^{ICE}$ with outcome models misspecified, and $\widehat\mu_t^{TMLE}$ with treatment models, outcome models, or both misspecified. All misspecified models omitted the term for $W_{it2}^2$ at each time $t$. 
    \subsection{Simulation results}
    Table \ref{tbl:sim} shows estimates of the bias, variance, and p-values from a Lilliefors test for normality, for each estimator of $\mu_5$. The results suggest that all stated theoretical properties hold approximately in simulated data. First, when all models are correctly specified, all estimators appear approximately unbiased with decreasing variance as the sample size increases. When outcome models and treatment models are misspecified, ICE and IPTW estimators appear biased, respectively. TMLE appears consistent when either the treatment or outcome models are correctly specified, but not when both are misspecified, supporting the double robustness property. Lastly, all estimators appear normally distributed for all sample sizes considered, based on Lilliefors tests. \\
 
\begin{table}[]
    \caption{Simulation results}
    \label{tbl:sim}
    \centering
    \begin{tabular}{lrrrrrrrrr}
    \hline
    \multicolumn{1}{c}{ } & \multicolumn{3}{c}{$n=1,000$} & \multicolumn{3}{c}{$n=10,000$} & \multicolumn{3}{c}{$n=100,000$} \\
    \hline
    estimator & variance$^1/n$ & bias$^2$ & p$^3$ & variance$^1/n$ & bias$^2$ & p$^3$ & variance$^1/n$ & bias$^2$ & p$^3$\\
    \hline
         ice\_qfal & 4.3 & 1.15 &   0.68 & 4.6 & 1.10 & 0.81 & 4.5 & 1.16 & 0.76\\
    			ice\_true & 4.2 & -0.04 &  0.82 & 4.5 & -0.02 & 0.23 & 4.2 & 0.02 & 0.45\\
    			iptw\_gfal & 4.5 & 1.17 &  0.86 & 4.7 & 1.17 & 0.49 & 4.6 & 1.25 & 0.69\\
    			iptw\_true & 6.0 & -0.17 & 0.37 & 6.3 & -0.08 & 0.87 & 7.4 & -0.01 & 0.57\\
    			tmle\_bfal & 4.4 & 1.19 &  0.51 & 4.6 & 1.15 & 0.65 & 4.5 & 1.22 & 0.62\\
    			tmle\_gfal & 4.2 & -0.08 & 0.86 & 4.5 & -0.04 & 0.28 & 4.3 & 0.02 & 0.21\\
    			tmle\_qfal & 5.9 & -0.09 & 0.15 & 6.4 & -0.07 & 0.78 & 7.6 & -0.02 & 0.41\\
    			tmle\_true & 5.4 & -0.03 & 0.82 & 5.6 & 0.01 & 0.11 & 5.7 & 0.01 & 0.69\\
    \hline \\
    \end{tabular}
    	\begin{tablenotes}
    			\item $^1$Empirical variance of estimates over 1000 simulated datasets.  $^2$Multiplied by 100. $^3$ P-value for Lilliefors test against the null hypothesis of normality.
    			\item Abbreviations: ice=iterated conditional expectation, iptw=inverse probability of treatment weighted, tmle=targeted maximum likelihood, qfal=outcome models misspecified, gfal=treatment models misspecified, bfal=both sets of models misspecified, true=all models correctly specified.
    		\end{tablenotes}
\end{table}

\section{COVID-19 application}\label{sec:ex_appl}

    \subsection{Data}
    This section presents an analysis of the motivating example, introduced in Section \ref{sec:ex_intr}. Code and data are provided in the R package didgformula (see Supporting Information). State-level weekly mortality data come from the Centers for Disease Control and Prevention's National Death Index, and weekly counts of COVID-19 cases from the COVID-19 Data Repository at the Center for Systems Science and Engineering at Johns Hopkins University. Data on state-level stay-at-home orders come from the COVID-19 U.S. State Policy database.  Though the outcome variable of interest is an individual-level indicator of death in week $t$, this variable is not directly observed; instead the observed data represent counts of deaths occurring in each state. Let $s=1,2,...,43$ be a state index, and let $Y_{ist}$ be an indicator of mortality during week $t$ for the $i$th individual ($i=1,...,n_s)$ living in state $s$, where $n_s$ denotes the population size in state $s$, and $n=\sum_{s=1}^{43}n_s\approx309$ million. The observed outcome variable is $Y_{st}=\sum_{i=1}^{n_s}Y_{ist}$, the state-level weekly sum of individual-level mortality counts, along with population counts $n_s$ (drawn from the 2010 Census). The observed treatment variable $A_{st}$ is an indicator of state $s$ being under stay-at-home order in week $t$. Finally, let $W_{st}$ be the change in confirmed COVID-19 cases reported per 100k population in the previous four weeks (i.e., the difference from week $t-4$ to $t$) in state $s$. Thus, in this example, the parallel trends assumption is conditional on the local state of the pandemic,  which may be plausible for pandemic-related policies \citep{callaway2021policy}. 
    
\subsection{Estimator implementation}
    
\subsubsection{IPTW}
For the treatment models, the following parametric models pooled over $k=1,...,11$ were assumed:
\begin{align*}
        f(A_{sk} | \overline A_{s,k-1}; \alpha_0) &=\Ber\{\logit^{-1}(\alpha_{00} + \alpha_{01} \omega(k) + \alpha_{02}A_{s,k-1})\} \\
        f(A_{sk} | \overline A_{s,k-1}, \overline W_{sk}; \alpha_1) &=  \Ber\{\logit^{-1}(\alpha_{10} + \alpha_{11} \omega(k) + \alpha_{12}A_{s,k-1} + \alpha_{13}\log W_{sk})\}
\end{align*}
where $\omega(k)$ is a natural cubic spline basis with 3 degrees of freedom for time $k$. The outcome model $c_{jk}(\overline A)=\gamma_{0jk}+\gamma_{1jk}I(\overline A_k=\overline 1)$, $k=1,...,11, j=k,k-1$ was specified, which allows the outcome to depend on the full exposure history. The parameters $\alpha_0=(\alpha_{00},\alpha_{01},\alpha_{02}) $ and $\alpha_1=(\alpha_{10}, ...,\alpha_{13})$ were estimated using maximum likelihood, weighted by $1/{n_s}$ to account for differing population sizes across states. Then, $\gamma_{0jk}, \gamma_{1jk}, k=1,...11, j=k,k-1$ were estimated by maximizing the state-level binomial likelihood weighted by inverse probability of treatment weights $\pi_k (\overline A; \overline W, \widehat \alpha)=\prod_{m=1}^k f(A_m|\overline A_{m-1}; \widehat\alpha_0 )/\prod_{m=1}^k f(A_m|\overline A_{m-1}, \overline W_m; \widehat\alpha_1)$, where $\widehat\alpha_0$ and $\widehat\alpha_1$ denote maximum likelihood estimators  of $\alpha_0$ and $\alpha_1$. Then estimators $\widehat\psi_t^{IPTW}, t=0,...,11$ were calculated as $\widehat\phi_{0,0}^{IPTW}+\sum_{k=1}^t\big( \widehat\phi_{k,k}^{IPTW} - \widehat\phi_{k-1,k}^{IPTW}\big)$, where $\widehat\phi_{j,k}^{IPTW}=\widehat\gamma_{0jk}+\widehat\gamma_{1jk}$ and $\widehat\gamma_{0jk},\widehat\gamma_{1jk}$ denote the weighted maximum likelihood estimators.

\subsubsection{ICE}
For ICE estimators, the following parametric outcome regression models pooled over $k=1,...,11$ were assumed:
\begin{align*}
    Q^{j,k,m}(\overline a^*; \beta_m) = \logit^{-1}\{\beta_{0jm} + \beta_{1jm}\omega(k) + \beta_{2jm}\omega(k)a_m^* +\beta_{3jm}\log W_{m} \}
\end{align*}
for $j=k,k-1$ and $m=k,k-1,...,0$, where again $\omega(k)$ refers to a natural cubic spline basis with 3 degrees of freedom. Note that, due to the monotonic treatment pattern, the interaction between time and treatment allows the outcome to depend on the full exposure history. The parameters $\beta_m$ were estimated by maximizing a binomial quasilikelihood, with estimators denoted  $\widehat\beta_m$. To account for varying state population sizes, state contributions to the quasilikelihood were weighted by $1/{n_s}$.  Finally, ICE estimators  $\widehat\psi_t^{ICE}, t=0,...,11$ were calculated as $\widehat\phi_{0,0}^{ICE}+\sum_{k=1}^t \big(\widehat\phi_{k,k}^{ICE} - \widehat\phi_{k-1,k}^{ICE}\big)$, where $\widehat\phi_{j,k}^{ICE}= \sum_{r=1}^{43} Q_r^{j,k,0}(\overline a^*; \widehat\beta_m) /43$ (as there are 43 states included in the analysis).

\subsubsection{TMLE}
For TMLE, the same treatment models as specified for IPTW were used, along with  the same outcome models as specified for ICE. Specifically, when estimating $\phi_{jk},$ for the $m$th ICE step ($m=k,k-1,...,0$), the TMLE updating step was performed by maximizing another weighted quasibinomial likelihood with response variable $Q_s^{j,k,m+1}(\overline a^*; \beta_{m+1})$ with an intercept and offset $Q_s^{j,k,m}(\overline a^*; \widehat\beta_m)$, weighted by $I(\overline A_k = \overline 1)/g_k(\overline A, \widehat\alpha_k).$ Predictions $\widehat Q_s^{j,k,m,*}(\overline a^*)$ from this model were then passed to the $(m-1)$th ICE step, and the process was repeated for $m=k,k-1,...,0$. Finally, $\widehat\psi_t^{TMLE}, t=0,...,11$ were calculated as $\widehat\phi_{0,0}^{TMLE}+\sum_{k=1}^t \big(\widehat\phi_{k,k}^{TMLE} - \widehat\phi_{k-1,k}^{TMLE}\big)$, where $\widehat\phi_{j,k}^{TMLE}= \sum_{r=1}^{43} \widehat Q_r^{j,k,0,*}(\overline a^*) /43.$

\subsubsection{Bootstrap standard errors and confidence intervals}
Standard errors were estimated using a nonparametric bootstrap. Specifically, for $B$ bootstrap replicates ($b=1,2,...,B$), a resampled outcome variable $Y_{st}^b=\sum_{i=1}^{n_s}Y_{ist}^b$ $(t=0,...,12)$ was drawn from a multinomial distribution with $n_s$ trials and probabilities $n_s^{-1}(Y_{s0}, Y_{s1}, ..., Y_{s,12}),$ where $Y_{s,12}$ denotes the number of individuals who survived beyond $t=11$ in state $s$. IPTW, ICE, and TMLE estimators $\widehat\psi_t^{IPTW,b},\widehat\psi_t^{ICE,b},\widehat\psi_t^{TMLE,b}$ were calculated on each replicate $(b=1,...,B$).  Then, Wald 95\% confidence intervals were computed using the standard deviation of bootstrap estimates.

\subsection{Results}
Figure \ref{fig:est_rates} shows results in the form of estimated U.S. weekly mortality rates per 100,000 person weeks over the study period under the natural course (red) and under the hypothetical sustained treatment of setting $A_t=1$ for all $t$, i.e., under a scenario where all 43 included states maintained stay-at-home orders through June 2020. The three estimators largely agree in their predictions that all-cause mortality rates would have been moderately lower throughout most of the study period, had stay-at-home orders remained in place. Translating the counterfactual mortality rate estimates to lives saved, if all causal and modeling assumptions hold, based on TMLE, stay-at-home orders remaining in place from April through June 2020 would have saved appoximately 11,100 (95\% CI: 6,800, 15,500) lives in those 43 states during the same time period. Results based on ICE were similar (point estimate: 11,300, 95\% CI: 6,900, 15,600), whereas IPTW gave a smaller point estimate and somewhat wider CI (point estimate: 4,100, 95\% CI: -500, 8,700). 

\begin{figure}[h]
    \centering
    \includegraphics{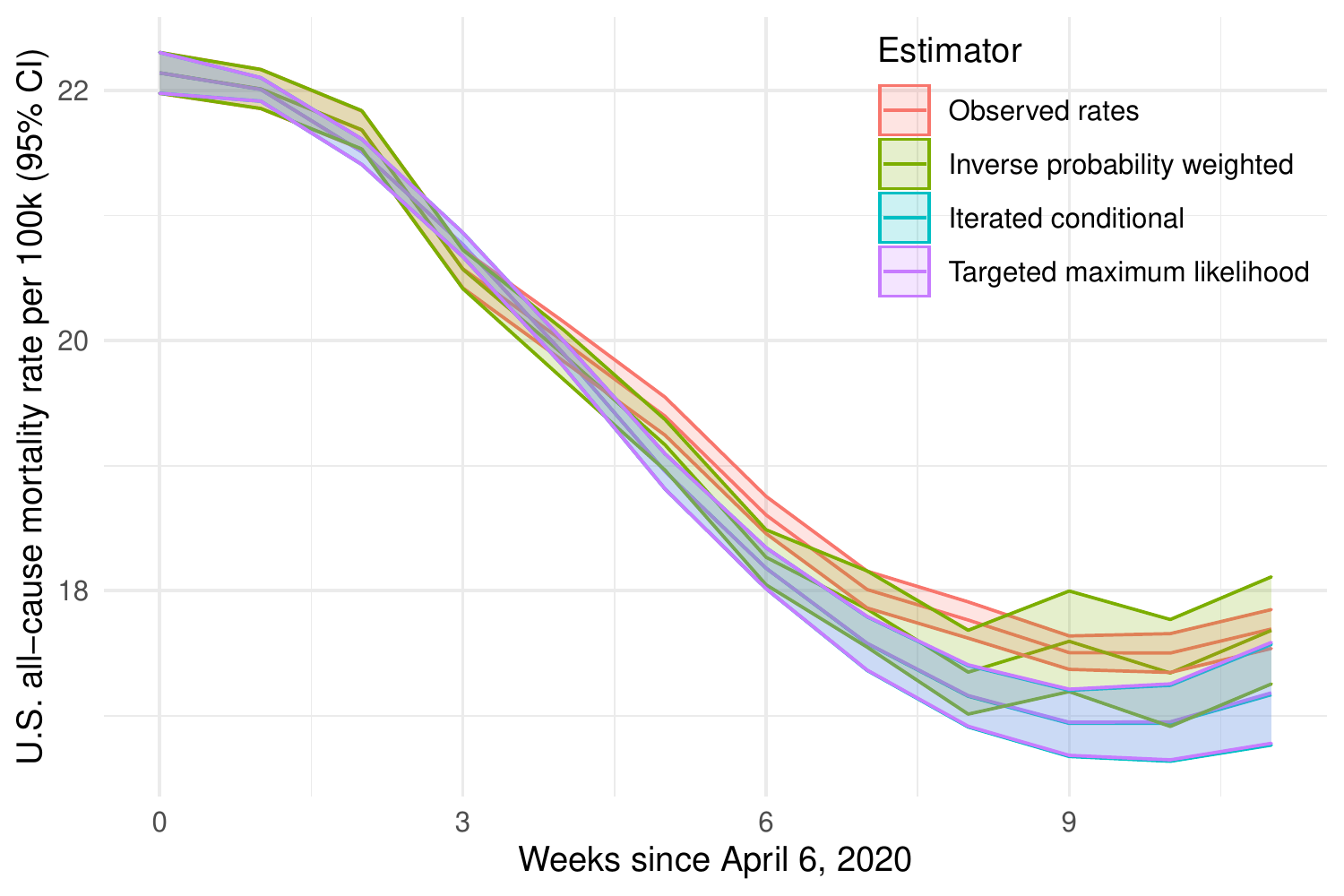}
    \caption{Estimated U.S. weekly mortality rates - observed (red) and estimated under hypothetical treatment setting all states to remain under stay-at-home order using IPTW (green), ICE (blue), and TMLE (purple). Note that TMLE and ICE estimates and 95\% CIs are nearly identical.}
    \label{fig:est_rates}
\end{figure}

\section{Extensions}
\subsection{Violations of parallel trends}
    In some applications, the parallel trends assumption (Assumption \ref{asn:pt}) may be questionable, and investigators may be interested in how inferences are altered by plausible deviations from parallel trends. A sensitivity analysis can be conducted as follows. Let
    \[
\Delta(\overline w_k, t) = \E\{Y_t(\overline a^*)-Y_{t-1}(\overline a^*)|\overline W_k=\overline w_k, \overline A_{k-1}=\overline a^*_{k-1}\}-\E\{Y_t(\overline a^*)-Y_{t-1}(\overline a^*)|\overline W_k=\overline w_k, \overline A_k=\overline a^*_k\},
    \]
    where $\Delta(\overline w_k, t)$ quantifies a deviation from parallel trends, which may depend on both the covariates $\overline w_k$ and time $t$. Then, consider the following statistical parameter:  
    \[
    \psi_t' = \E(Y_0) + \sum_{k=1}^t \int \bigg\{\E(Y_k-Y_{k-1}|\overline{A}_k=\overline a^*_k, \overline W_k=\overline w_k\} + \sum_{m=1}^k \Delta(\overline w_m, k) \bigg\} \prod_{m=0}^k dF(w_m|\overline w_{m-1}, \overline a_{m-1}^*).
    \]
    If Assumptions \ref{asn:sutva} and \ref{asn:pos} hold, then $\mu_t=\psi_t'$ (the proof follows from results in Appendix A). If a particular value is assumed known for $\Delta(\overline w_k, t)$, then estimation can proceed by defining
    \[
    \phi_{j,k}'=\int \bigg[\E\bigg\{Y_j + \sum_{m=1}^k \Delta(\overline w_m, k)|\overline A_t=\overline a_t^*, \overline W_k=\overline w_k\bigg\}  \bigg]\prod_{k=0}^t dF(w_k|\overline  W_{k-1}=\overline w_{k-1}, \overline A_{k-1}=\overline a_{k-1}^*),   
    \]
    and noting that $\psi_t' = \phi_{0,0} + \sum_{k=1}^t (\phi_{k,k}'-\phi_{k-1,k})$. As with $\phi_{j,k}$, the parameter $\phi_{j,k}'$ is simply a special case of the usual g-formula, where the outcome variable is $Y_j + \sum_{m=1}^k \Delta(\overline w_m, k)$. Thus, the IPTW, ICE, and TMLE estimators can be used, replacing outcome variables $Y_k$ with $Y_k + \sum_{m=1}^k \Delta(\overline w_m, k)$ (but not for $Y_{k-1}$). In practice, $\Delta(\overline w_k, t)$ will typically not be known, and thus estimates may be computed over a range of plausible values of $\Delta(\overline w_k, t)$. Differences in trends between subgroups of units before discontinuation occurs may be helpful in determining plausible values of $\Delta(\overline w_k, t)$ \citep{roth2019honest}.
\subsection{Dynamic regimes}
In addition to the static regimes considered above, the proposed approach can accommodate regimes where treatment decisions may depend on the history of covariates and/or treatments. Let $\overline g = \{g_0(w_0), g_1(\overline w_1), ..., g_{\tau}(\overline w_{\tau}) \}$ denote a dynamic regime, where $g_k(\overline w_k)$ returns the treatment value $a_k$ that would be assigned given covariate history $\overline w_k$. Note that $g_k(\cdot)$ may depend on treatment history as well, which we suppress for notational simplicity. Likewise let $Y_k(\overline g)$ be a potential outcome under treatment regime $\overline g$. Suppose interest is in the estimand $\mu_t^g = \E\{Y_t(\overline g)\}.$ Then, consider the following modifications to Assumptions 1-3.
\begin{assumption}\label{asn:sutva_dyn} (SUTVA for dynamic regimes): If $\overline A_{it} = \overline g_t(\overline W_{it})$, then $Y_{it}=Y_{it}(\overline g_t)$ for $t \in \{0,1,...\tau\}$.
\end{assumption}
\begin{assumption}\label{asn:pos_dyn} (Positivity for dynamic regimes): If $f\{\overline w_t|\overline A_{t-1}=\overline g_{t-1}(\overline W_{t-1})\}>0,$ then $f\{g_t(\overline w_t)|\overline W_t=\overline w_t,\overline A_{t-1}=\overline  g_{t-1}(\overline w_{t-1}) \}>0$, for $\overline w_t \in \mathcal{\overline{W}}_t; t\in\{1,2,...,\tau\}$.
\end{assumption}
\begin{assumption}\label{asn:pt_dyn} (Parallel trends for dynamic regimes): For $t \in \{1,2,...,\tau\}, k\leq t:$
    	\[
    	\E\{Y_t(\overline g)-Y_{t-1}(\overline g)|\overline W_k, \overline A_{k-1}=\overline g_{k-1}(\overline W_{k-1})\}=\E\{Y_t(\overline g)-Y_{t-1}(\overline g)|\overline W_k, \overline A_k=\overline g_k(\overline W_k)\} \]
\end{assumption}
\begin{lemma}\label{lem:gform_dyn} (Parallel trends g-formula, dynamic regimes) Define the functional (i.e., statistical parameter)
    \[\psi_t^g \equiv \E(Y_0)+\sum_{k=1}^t \int \E\{Y_k-Y_{k-1}|\overline W_k =\overline w_k,  \overline A_k =\overline g_k(\overline w_k)\}\prod_{m=0}^k dF\{w_m|\overline w_{m-1},  \overline g_{m-1}(\overline w_{m-1})\}\] Under a staggered discontinuation design and if Assumptions \ref{asn:sutva_dyn}-\ref{asn:pt_dyn} hold, then $\psi_t^g = \mu_t^g$.
\end{lemma}
The proof of Lemma \ref{lem:gform_dyn} follows from results in Appendix A. Thus, the IPTW, ICE, and TMLE estimators described can be used, with $\phi_{j,k}$ appropriately redefined.

\section{Discussion}
This paper considers a new approach to identifying effects of sustained intervention strategies based on an assumption set that includes parallel trends. This assumption is popular in difference-in-differences because it allows for some degree of unmeasured confounding \citep{Zeldow2019}. Recently, parallel trends assumptions have been leveraged to target sustained treatment estimands, mainly considering certain types of treatment regimes \citep{callaway2021difference, chaise2020hetero, chaise2021several, chaise2021inter}. Relative to previous work, the main contribution of this paper is a framework for estimating marginal intervention-specific means for general treatment regimes (including dynamic regimes) under parallel trends. This is accomplished by building on IPTW, g-computation, and doubly-robust TMLE developed in the context of sequential exchangeability,   thus connecting disparate causal inference literatures from biostatistics \citep{Robins1986,Robins2000,Bang2005,Vanderlaan2012} and econometrics \citep{Ashenfelter1985, callaway2021difference}. Independently and concurrently with the present work, \citet{Shahn2022} developed g-estimation of stuctural nested models for general sustained treatment regimes under parallel trends, with results that imply identification for the intervention-specific means considered here. While it is possible (but complex) to estimate the latter quantity using the g-estimation approach of \citet{Shahn2022}, the main strength of g-estimation is in exploring effect heterogeneity by time-varying covariates.

Regarding the example presented in Section \ref{sec:ex_appl}, care should be taken when assuming parallel trends for pandemic-related outcomes without conditioning on pandemic state variables such as infection rates, as marginal parallel trends are incompatible with standard epidemic models \citep{callaway2021policy}. DID methods have previously been used to estimate effects of stay-at-home orders on the treated \citep[e.g.][]{fowler2021stay}. The methods in this paper allow for (i) a different target parameter that may more directly correspond to decisions facing policy makers and public health officials \citep{maldonado2002estimating}, and (ii) adjustment for time-varying pandemic state variables likely affected by prior treatment, which DID methods have only recently begun to consider \citep{callaway2021policy}. That said, assessing the effects of stay-at-home orders is complex, and a comprehensive analysis would need to consider potential biases not factored into the present analysis; e.g., there is likely some interference \citep{Haber2021}. Thus, the application results are not meant to inform policy or scientific conclusions.
    
The approach presented here may have application in many other contexts. Many U.S. state-level policies have changed in such a way as to accommodate a staggered discontinuation design, including in domains other than pandemic mitigation. Outside of staggered discontinuation designs, methods developed in this paper apply more generally in settings where baseline potential outcomes are identified.  For example, the approach could be used to estimate per-protocol effects in a clinical trial of a time-varying treatment regime with non-adherence.

Several areas for future research remain. First, it will be important to explore efficiency for competing estimators in this framework. Notably, the TMLE presented here is only known to be semiparametric efficient for the nuisance parameters $\phi_{jk}$ and not necessarily for the target parameter $\psi_t$ \citep{Vanderlaan2012}. Second, though parallel trends may be considered more plausible than sequential exchangeability in some settings, strategies for formally evaluating the assumption using domain knowledge, e.g. using causal diagrams, are in their infancy \citep[e.g.,][]{ghanem2022selection}. Finally, the focus of this paper was on settings where one treatment regime is of interest. Future research could consider extensions beyond a single regimen, but caution should be exercised when assuming parallel trends for multiple regimens, which would imply certain restrictions on treatment effect heterogeneity that may not be plausible in some settings \citep{Shahn2022}.


\section*{Acknowledgements}
The authors thank Dr. Whitney Robinson for helpful comments. This research was supported by the NIH grants T32-HD091058-02, T32-AI007001, R01 AI085073, and P2C-HD050924. The content is solely the responsibility of the authors and does not represent the official views of the National Institutes of Health.

\section*{Software}
The R package \verb|didgformula|, available at https://github.com/audreyrenson/didgformula, implements the estimators, simulation study (in vignette ``simulation''), and example results (in vignette ``example'') described in the paper. 

\section*{Data availability statement}
The data that support the findings in this paper are available as part of the R package \verb|didgformula|, available at https://github.com/audreyrenson/didgformula, in the dataset called \verb|stayathome2020|.

\begin{appendices}
\renewcommand*{\thesection}{\Alph{section}}
\setcounter{section}{0}
\renewcommand*{\theequation}{\Alph{section}\arabic{equation}}
\setcounter{equation}{0}
\section{Proof of Lemma 1}\label{app1}
	Here we provide a formal proof by induction of Lemma 1. 
	\begin{proof}
	First, note that, by adding and subtracting constants,
	\[
	\E[Y_t(\bar a^*)] = \E[Y_0(\bar a^*)] + \sum_{k=1}^t \E[Y_k(\bar a^*) - Y_{k-1}(\bar a^*)],
	\]
	and by SUTVA (Assumption 1), $ \E[Y_0(\bar a^*)]= \E[Y_0]$ under a staggered discontinuation design because $A_{i0}=a_0^*$ for all $i$. Thus it remains to prove that
	\[
	\E[Y_k(\bar a^*) - Y_{k-1}(\bar a^*)] = \int \E[Y_k - Y_{k-1}|\bar W_k=\bar w_k, \bar A_k=\bar a_k^*]\prod_{m=0}^{k}dF(w_m|\bar w_{m-1}, \bar a^*)
	\]
	for $1 \leq k \leq t$. As our induction hypothesis, assume temporarily that the following holds for some $m$ such that $1\leq m < k \leq t$:
	\begin{equation}\label{ind_hyp_pt}
		\E[Y_k(\bar a^*) - Y_{k-1}(\bar a^*)] = \int \E[Y_k(\bar a^*) - Y_{k-1}(\bar a^*)|\bar W_m=\bar w_m, \bar A_m=\bar a_m^*]\prod_{s=0}^m dF(w_s|\bar w_{s-1}, \bar a^*_{s-1})
	\end{equation}
	If (\ref{ind_hyp_pt}) is true, then, so long as $m\leq t-1$, it follows that:
	\begin{align}\label{ind_impl_pt}
		\nonumber \E[Y_k(\bar a^*) - Y_{k-1}(\bar a^*)] &= \int \E[Y_k(\bar a^*) - Y_{k-1}(\bar a^*)|\bar W_{m+1}=\bar w_{m+1}, \bar A_m=\bar a^*_m]\prod_{s=0}^{m+1} dF(w_s|\bar w_{s-1}, \bar a^*_{s-1}) \\
		&= \int \E[Y_k(\bar a^*) - Y_{k-1}(\bar a^*)|\bar W_{m+1}=\bar w_{m+1}, \bar A_{m+1}=\bar a^*_{m+1}]\prod_{s=0}^{m+1} dF(w_s|\bar w_{s-1}, \bar a^*_{s-1}) 
	\end{align}
	where the first equality is by iterated expectation and the second by Assumption 3. In other words, if (\ref{ind_hyp_pt}) holds for some $m$, (\ref{ind_impl_pt}) proves that the same statement holds for $m+1\leq t$. Next, note that, for all $k$ such that $1\leq k \leq t$, by iterated expectation and Assumption 3:
	\begin{align*}
		\E[Y_k(\bar a^*) - Y_{k-1}(\bar a^*)] &= \int \E[Y_k(\bar a^*) - Y_{k-1}(\bar a^*)|\bar W_1=\bar w_1]\prod_{s=0}^1 dF(w_s | w_{s-1}, a_{s-1}^*) \\
		&= \int \E[Y_k(\bar a^*) - Y_{k-1}(\bar a^*)|\bar W_1=\bar w_1, \bar A_1=\bar a_1^*]\prod_{s=0}^1 dF(w_s | w_{s-1},a_{s-1}^*) 
	\end{align*}
	which shows that (\ref{ind_hyp_pt}) holds for $m=1$. Then, by (\ref{ind_impl_pt}), (\ref{ind_hyp_pt}) holds for $m=2,...,k$. Finally, for $m=k$, by Assumption 1 we have:
	\begin{align*}
		\E[Y_t(\bar a^*) - Y_{t-1}(\bar a^*)] &= \int \E[Y_t(\bar a^*) - Y_{t-1}(\bar a^*)|\bar W_k=\bar w_k, \bar A_k=\bar a_k^*]\prod_{s=0}^k dF(w_s | \bar w_{s-1}, \bar a_{s-1}^*) \\
		&= \int \E[Y_t - Y_{t-1}|\bar W_k=\bar w_k, \bar A_k=\bar a_k^*]\prod_{s=0}^k dF(w_s | \bar w_{s-1}, \bar a_{s-1}^*) 
	\end{align*}
    \end{proof}
    
\setcounter{equation}{0}    
    
\section{Proof of parallel trends in simulation data generating distribution}\label{app2}
To prove that parallel trends holds under the simulation setup, we first introduce three additional assumptions which trivially hold in the simulation, and then show that they are sufficient to guarantee parallel trends.
\begin{assumption}\label{asn:aec} (Additive equi-confounding) For all $u, u' \in \mathcal{U}_0$,
\begin{align*}
    \E[Y_t(\bar a^*)&|\bar W_t, \bar A_t=\bar a_t^*, U_0=u] - E[Y_t(\bar a^*)|\bar W_t, \bar A_t=\bar a_t^*, U_0=u']\\&=E[Y_{t-1}(\bar a^*)|\bar W_t, \bar A_t=\bar a_t^*, U_0=u] - E[Y_{t-1}(\bar a^*)|\bar W_t, \bar A_t=\bar a_t^*, U_0=u'] \\
\end{align*}   
\end{assumption}
\begin{assumption}\label{asn:ign} (Latent ignorability)
	\[
	(Y_t(\bar{a}_t),	Y_{t-1}(\bar{a}_t) )  \ind A_k | U_0, \bar W_k, \bar A_{k-1}=\bar a_{k-1}^*\\
	\text{for } 1 \leq t \leq T  \text{, } k \leq t \]
\end{assumption}
\begin{assumption}\label{asn:cov_ind_u} (Measured covariates conditionally independent of unmeasured ones)
	\[
	W_t \ind U_0 | \bar W_{t-1}, \bar A_{t-1} =\bar a_{t-1}^* 
	\]
	\[
	\text{for } 1 \leq t \leq T \]
\end{assumption}
\begin{lemma}\label{lemma:implied_pt} (Implied parallel trends) Under Assumptions \ref{asn:aec}-\ref{asn:cov_ind_u},
	\[
	\E[Y_t(\bar a^*)-Y_{t-1}(\bar a^*)|\bar W_k, \bar A_{k-1}=\bar a_{k-1}^*]=\E[Y_t(\bar a^*)-Y_{t-1}(\bar a^*)|\bar W_k, \bar A_k=\bar a_k^*] \text{ for } t\in \{1,2,...,\tau\}, k\leq t.
	\]
\end{lemma}
\begin{proof}
\textit{Proof of Lemma \ref{lemma:implied_pt}:} To show that parallel trends (Assumption 3) are implied by Assumptions \ref{asn:aec}-\ref{asn:cov_ind_u}, we show both that
\begin{align}\label{temp2}
\nonumber\E[&Y_t(\bar{a}^*) - Y_{t-1}(\bar{a}^*)| \bar{A}_k=\bar{a}_k^*, \bar W_k]\\&=
\E_{W_{k+1}|\bar{A}_k=\bar a_k^*, \bar W_k}\bigg[
\cdot\cdot\cdot\E_{W_t|\bar A_{t-1} = \bar a_{t-1}^*, \bar W_{t-1}} \bigg\{
\E[Y_t(\bar{a}^*) - Y_{t-1}(\bar{a}^*) | \bar{A}_t=\bar{a}_t^*, \bar W_t]\bigg\}\cdot\cdot\cdot\bigg]
\end{align}
and
\begin{align}\label{temp3}
\nonumber\E[&Y_t(\bar{a}^*) - Y_{t-1}(\bar{a}^*)| \bar{A}_{k-1}=\bar{a}_{k-1}^*, \bar W_k]\\&=
\E_{W_{k+1}|\bar{A}_k=\bar a_k^*, \bar W_k}\bigg[
\cdot\cdot\cdot\E_{W_t|\bar A_{t-1} = \bar a_{t-1}^*, \bar W_{t-1}} \bigg\{
\E[Y_t(\bar{a}^*) - Y_{t-1}(\bar{a}^*) | \bar{A}_t=\bar{a}_t^*, \bar W_t]\bigg\}\cdot\cdot\cdot\bigg]
\end{align}
To see (\ref{temp2}) note that, for any $t,k$ such that $1\leq k\leq t\leq T$, by repeated applications of iterated expectation and Assumption \ref{asn:ign} we have:
\begin{align*}
&\E[Y_t(\bar{a}^*) - Y_{t-1}(\bar{a}^*)| \bar{A}_k=\bar{a}_k^*, \bar W_k] \\
&\small=\E_{U_0|\bar{A}_k=\bar{a}_k^*, \bar W_k} \bigg(
	\E_{W_{k+1}|\bar{A}_k=\bar a_k^*, \bar W_k, U_0}\bigg[
	\cdot\cdot\cdot\E_{W_t|\bar A_{t-1} = \bar a_{t-1}^*, \bar W_{t-1}, U_0} \bigg\{
	\E[Y_t(\bar{a}^*) - Y_{t-1}(\bar{a}^*) | \bar{A}_t=\bar{a}_t^*, \bar W_t, U_0]\bigg\}\cdot\cdot\cdot\bigg]\bigg)
\end{align*}
By Assumption \ref{asn:aec}:
\[
=\E_{U_0|\bar{A}_k=\bar{a}_k^*, \bar{W}_k} \bigg(
\E_{W_{k+1}|\bar{A}_k=\bar a_k^*, \bar W_k, U_0}\bigg[
\cdot\cdot\cdot\E_{W_t|\bar A_{t-1} = \bar a_{t-1}^*, \bar W_{t-1}, U_0} \bigg\{
\E[Y_t(\bar{a}^*) - Y_{t-1}(\bar{a}^*) | \bar{A}_t=\bar{a}_t^*, \bar{W}_t]\bigg\}\cdot\cdot\cdot\bigg]\bigg)
\]
By Assumption \ref{asn:cov_ind_u}:
\[
=\E_{U_0|\bar{A}_k=\bar{a}_k^*, \bar{W}_k} \bigg(
\E_{W_{k+1}|\bar{A}_k=\bar a_k^*, \bar W_k}\bigg[
\cdot\cdot\cdot\E_{W_t|\bar A_{t-1} = \bar a_{t-1}^*, \bar W_{t-1}} \bigg\{
\E[Y_t(\bar{a}^*) - Y_{t-1}(\bar{a}^*) | \bar{A}_t=\bar{a}_t^*, \bar{W}_t]\bigg\}\cdot\cdot\cdot\bigg]\bigg)
\]
\begin{equation}\label{temp4}
=\E_{W_{k+1}|\bar{A}_k=\bar a_k^*, \bar W_k}\bigg[
\cdot\cdot\cdot\E_{W_t|\bar A_{t-1} = \bar a_{t-1}^*, \bar W_{t-1}} \bigg\{
\E[Y_t(\bar{a}^*) - Y_{t-1}(\bar{a}^*) | \bar{A}_t=\bar{a}_t^*, \bar{W}_t]\bigg\}\cdot\cdot\cdot\bigg]
\end{equation}
The last equality holds because the inside terms do not depend on $U_0$. To see (\ref{temp3}), repeated applications of iterated expectation and Assumption \ref{asn:ign} again give:
\begin{align*}
	&\E[Y_t(\bar{a}^*) - Y_{t-1}(\bar{a}^*)| \bar{A}_{k-1}=\bar{a}_{k-1}^*, \bar{W}_k] \\
	&\tiny{=\E_{U_0|\bar{A}_{k-1}=\bar{a}_{k-1}^*, \bar{W}_k} \bigg(
	\E_{W_{k+1}|\bar{A}_k=\bar a_k^*, \bar W_k, U_0}\bigg[
	\cdot\cdot\cdot\E_{W_t|\bar A_{t-1} = \bar a_{t-1}^*, \bar W_{t-1}, U_0} \bigg\{
	\E[Y_t(\bar{a}) - Y_{t-1}(\bar{a}) | \bar{A}_t=\bar{a}_t^*, \bar{W}_t, U_0]\bigg\}\cdot\cdot\cdot\bigg]\bigg)}
\end{align*}
Similarly, by Assumptions \ref{asn:aec} and \ref{asn:cov_ind_u}:
\[
=\E_{U_0|\bar{A}_k=\bar{a}_k^*, \bar{W}_k} \bigg(
\E_{W_{k+1}|bar{A}_{k-1}=\bar{a}_{k-1}^*, \bar W_k}\bigg[
\cdot\cdot\cdot\E_{W_t|\bar A_{t-1} = \bar a_{t-1}^*, \bar W_{t-1}} \bigg\{
\E[Y_t(\bar{a}^*) - Y_{t-1}(\bar{a}^*) | \bar{A}_t=\bar{a}_t^*, \bar{W}_t]\bigg\}\cdot\cdot\cdot\bigg]\bigg)
\]
\[
=
\E_{W_{k+1}|\bar{A}_k=\bar a_k^*, \bar W_k}\bigg[
\cdot\cdot\cdot\E_{W_t|\bar A_{t-1} = \bar a_{t-1}^*, \bar W_{t-1}} \bigg\{
\E[Y_t(\bar{a}) - Y_{t-1}(\bar{a}) | \bar{A}_t=\bar{a}_t^*, \bar{W}_t]\bigg\}\cdot\cdot\cdot\bigg]
\]
which is exactly equal to (\ref{temp4}), thus proving that parallel trends (Assumption 3) hold whenever Assumptions \ref{asn:aec}-\ref{asn:cov_ind_u} hold. This ends the proof of Lemma \ref{lemma:implied_pt}.\\
\end{proof}

Thus, since Assumptions \ref{asn:aec}-\ref{asn:cov_ind_u} hold in the simulation data-generating distribution, parallel trends also holds. In particular, Assumptions \ref{asn:ign}-\ref{asn:cov_ind_u} are easy to see in the simulation setup, and Assumption \ref{asn:aec} holds because:
\begin{align*}
    \E[Y_t&(\bar a^*)|\bar W_t, \bar A_t=\bar a_t^*, U_0=u] - \E[Y_t(\bar a^*)|\bar W_t, \bar A_t=\bar a_t^*, U_0=u'] \\
    &= \E[Y_{t-1}(\bar a^*)|\bar W_t, \bar A_t=\bar a_t^*, U_0=u] - \E[Y_{t-1}(\bar a^*)|\bar W_t, \bar A_t=\bar a_t^*, U_0=u']) \\
    &=\theta (u-u')
\end{align*}
where $\theta$ is the constant (over $t$) linear model coefficient relating $Y_t$ to $U_0$ in the simulation data generating distribution described in Section 5.1.\\

\end{appendices}

\end{document}